\documentclass[12pt,reqno]{amsart} 

\usepackage{mathdots}
\usepackage{amsmath, amssymb} %% AMS symbol definitions
\usepackage{amscd}            %% for Commutative Diagrams
\usepackage{amsxtra}          %% defines \spcheck, \spdot, etc.
\usepackage{upref}   
\usepackage{enumerate}         %% cross-refs to sections are upright

\usepackage[mathscr]{eucal}

\theoremstyle{plain}

\newtheorem{theorem}{Theorem}

\newtheorem{prop}[theorem]{Proposition}

\newtheorem{theo}[theorem]{Theorem}

\newtheorem{lemma}[theorem]{Lemma}

\newtheorem{definition}[theorem]{Definition}
\newtheorem{example}[theorem]{Example}
\theoremstyle{remark}
\newtheorem{remark}[theorem]{Remark}

\numberwithin{theorem}{section}
\newcommand{\be}%
  {\protect\setcounter{equation}{\value{subsubsection}}}  
\newcommand{\ee}%
\newcommand{\Z}{\mathbb{Z}}

\newcommand{\F}{\mathbb{F}}

\newcommand{\xn}{x^n - 1}

\begin{document}

%%%%%%%%%%%%%%%%%%%%%%%%%%%%%%%%%%%%%%%%%%%%%%%%%%%%%%%%%%%%%%%%%%%%%
%%%%%%%%%%%%%%%%%%%%%%%%%%%%%%%%%%%%%%%%%%%%%%%%%%%%%%%%%%%%%%%%%%%%%
%% TOPMATTER
%%%%%%%%%%%%%%%%%%%%%%%%%%%%%%%%%%%%%%%%%%%%%%%%%%%%%%%%%%%%%%%%%%%%%

\title {Cyclic codes over the ring $ \Z_p[u, v]/\langle u^2, v^2, uv-vu\rangle$}
%%%%%

%%%%%
\author{Pramod Kumar Kewat, Bappaditya Ghosh and Sukhamoy Pattanayak} %% This is the correct form!
\address{Department of Applied Mathematics\\ 
         Indian School of Mines\\
         Dhanbad 826 004,  India}
\email{kewat.pk.am@ismdhanbad.ac.in\\bappaditya.ghosh86@gmail.com\\sukhamoy88@gmail.com}

\begin{abstract}
Let $p$ be a prime number. In this paper, we study cyclic codes over the ring $ \Z_p[u, v]/\langle u^2, v^2, uv-vu\rangle$. We find a unique set of generators for these codes. We also study the rank and the Hamming distance of these codes. We obtain all except one ternary optimal code of length 12 as the Gray image of the cyclic codes over the ring $ \Z_p[u, v]/\langle u^2, v^2, uv-vu\rangle$. We also characterize the $p$-ary image of these cyclic codes under the Gray map. 
\end{abstract}

\maketitle

%%    LEFT AND RIGHT RUNNING HEADS.  MUST COME *AFTER* \maketitle

\markboth{P.K. Kewat, B. Ghosh and S. Patnayak}{Cyclic codes over the ring $R_{u^2,v^2,p}$}
%%%%%%%%%%%%%%%%%%%%%%%%%%%%%%%%%%%%%%%%%%%%%%%%%%%%%%%%%%%%%%%%%%%%%
%% BODY OF PAPER BEGINS HERE
%%%%%%%%%%%%%%%%%%%%%%%%%%%%%%%%%%%%%%%%%%%%%%%%%%%%%%%%%%%%%%%%%%%%%
%%%%%%%%%%%%%%%%%%%%%%%%%%%%%%%%%%%%%%%%%%%%%%%%%%%%%%%%%%%%%%%%%%%%%
%% INTRODUCTION
%%%%%%%%%%%%%%%%%%%%%%%%%%%%%%%%%%%%%%%%%%%%%%%%%%%%%%%%%%%%%%%%%%%%%
\section{Introduction}
Let $R$ be a ring. A linear code of length $n$ over $ R$ is a $R$ submodule of $R^{n}$. A linear code $C$ of length $n$ over $R$ is cyclic if $ (c_{n-1}, c_0, \dots, c_{n-2}) \in C$ whenever $ (c_0, c_1, \dots, c_{n-1}) \in C$. We can consider a cyclic code $C$ of length $n$ over $R$ as an ideal in $R[x]/\langle x^n - 1\rangle$ via the following correspondence
\begin{center} $ R^{n} \longrightarrow R[x]/\langle x^n - 1\rangle, ~~ (c_0, c_1, \dots, c_{n-1})\mapsto c_0 + c_1x + \cdots + c_{n-1}x^{n-1}.$
\end{center}

The cyclic codes form an important class of codes due to their rich algebraic structures and play a significant role in algebraic coding theory. The structure of cyclic codes over rings has been discussed in a series of papers \cite{Tah-Oeh04, Tah-Oeh03, Blac03, Bonn-Udaya99, Cal-Slo95, Cal-Slo98, Con-Slo93,  Dou-Shiro01,Ples-Qian96, Lint91}. The structures of cyclic codes of length $n$ over a finite chain ring have been discussed in \cite{Dinh-Lopez04}  when $n$ is not divisible by the characteristic of the residue field $\bar{R}$. 
The structures of cyclic codes of length $n$, where $n$ is divisible by the characteristic of the residue field $\bar{R}$,  over some finite chain ring have been discussed in \cite{Tah-Siap07,Ash-Ham11,Dinh10,Aks-Pkk13}. 

Yildiz and Karadeniz in \cite{Yil-Kar11} have considered the ring $\F_2[u,v]/\langle u^2, v^2, uv-vu \rangle$, which is not a chain ring, and studied cyclic codes of odd length over that. They have find some good binary codes as the Gray images of these cyclic codes. The authors of \cite{Dou-Yil-Kar12} have considered the more general ring $\F_2[u_1, u_2, \cdots, u_k]/\langle u_{i}^2, u_{j}^2, u_iu_j-u_ju_i \rangle$, and studied the general properties of cyclic codes over these rings and  characterized the nontrivial one-generator cyclic codes. Sobhani and Molakarimi in \cite{Sob-Mor13} extended these studies to cyclic codes over the ring $\F_{2^m}[u, v]/\langle u^2, v^2, uv-vu \rangle$. 

\indent
Let $R_{u^2,v^2,p} = \Z_p + u\Z_p + v\Z_p+uv\Z_p$, where $u^2=0$, $v^2=0$, $ uv =vu$ and $p$ is a prime number. Note that the ring $R_{u^2,v^2,p}$ can also be viewed as the quotient ring $ \Z_p[u,v]/\langle u^2, v^2, uv-vu\rangle$. In this paper, we discuss the structure of cyclic codes of arbitrary length over the ring $R_{u^2,v^2, p}$. We find a unique set of generators, rank and a minimal spanning set for these codes. We also find the Hamming distance of these codes for length $p^l$. Note that the rank and the Hamming distance of these cyclic codes for $p=2$ have not been discussed in  \cite{Dou-Yil-Kar12, Sob-Mor13, Yil-Kar11}. Hence getting the rank and the Hamming distance of the cyclic codes over $R_{u^2,v^2,p} $ are new even for $p=2$.
We obtain all except one optimal code over $\Z_3$ of length 12 as the images of cyclic codes over $R_{u^2,v^2,p}$ under the Gray map. We also characterize the $p$-ary image of these cyclic codes under the Gray map.

%Recall that the Hamming weight of a codeword $c$ is defined as the number of non-zero entries of $c$ and the Hamming distance of a code $C$ is the shttp://www.ismdhanbad.ac.in/mallest possible weight among all its non zero codewords. The minimum distance of a code is the minimum Hamming distance between two distinct codewords. When the code is linear, the minimum distance of the code is equal to the Hamming distance. The minimum distance determines the maximum number of errors that can be corrected under any decoding algorithm.

\indent
Let $R_{u^2,p} = \Z_p +u\Z_p, u^2=0$ and $R_{u^2,p,n}= R_{u^2,p}/\langle\xn \rangle$. The line of arguments we have used to find a set of generators are similar to \cite{Tah-Siap07, Aks-Pkk13}. Let $C$ be a cyclic code over the ring $R_{u^2,v^2,p}$. The idea to find a set of generators is as follows. We view the cyclic code $C$ as an ideal in $R_{u^2,v^2,p,n}= R_{u^2,v^2,p}/\langle\xn \rangle$. Then we define the projection map from $R_{u^2,v^2,p,n} \longrightarrow R_{u^2,p,n}$ and we get an ideal in $R_{u^2,p,n}$, which gives a cyclic codes over $R_{u^2,p}$. The structure of cyclic codes over $R_{u^2,p}$  is known from \cite{Aks-Pkk13}. By pullback, we find a set of generators for a cyclic code over $R_{u^2,v^2,p}$. By using the division algorithm and direct computations, we have find the rank of these cyclic codes. The line of arguments we have used to find minimum distance are similar to \cite{Tah-Siap07, Aks-Pkk13} . \\
\
\indent
This paper is organized as follows. In Section \ref{pre}, we discuss the preliminaries that we need. In Section \ref{generator} we give a unique set of generators for the cyclic codes $C$ over the ring $R_{u^2, v^2,p}$. In Section \ref{rank}, we find a minimal spanning set for these codes and discuss about the rank. In Section \ref{md}, we find the minimum distance of these codes for length $p^l$. In Section \ref{exm}, we discuss some of the examples of these codes and give a list of cyclic codes $C$ of length 3 over $R_{u^2, v^2,3}$ whose ternary image gives optimal code under the Gray map. In Section \ref{p-ary-im}, we characterize the $p$-ary image of cyclic codes over $R_{u^2, v^2,p}$ under the Gray map.  

\section{Preliminaries}\label{pre}

%A ring is called a right (left) chain ring if the set of all right (left) ideals of $R$ is a chain under set-theoretic inclusion. If $R$ is both a right and a left chain ring, we simply call $R$ a chain ring.

%We have the following equivalent conditions for the finite commutative rings.
%\begin{prop}{\rm \cite[Proposition 2.1]{Dinh-Lopez04}}\label{chain-ring}
%The following conditions are equivalent for a finite commutative ring $R$.
%\begin{enumerate}[{\rm (1)}]
%\item  $R$ is a local ring and the maximal ideal $M$ of $R$ is principal;
%\item  $R$ is a local principal ideal ring;
%\item  $R$ is a chain ring.
%\end{enumerate}
%\end{prop}
A ring $R$ is called local ring if $R$ has a unique maximal ideal. Let $R$ be a finite commutative local ring with maximal ideal $M$. We define the residue field $\overline{R} = R/M$. Let $\mu : R[x] \rightarrow \overline{R}[x]$ denote the natural ring homomorphism that maps $r \mapsto r + M$ and the variable $x$ to $x$. We define the degree of the polynomial $f(x) \in R[x]$ as the degree of the polynomial $\mu(f(x))$ in $\overline{R}[x]$, i.e., $deg(f(x)) = deg(\mu(f(x))$ (see, for example, \cite{McDonald74}). A polynomial $f(x) \in R[x]$ is called regular if it is not a zero divisor.
The following conditions are equivalent for a finite commutative local ring $R$.
\begin{prop} {\rm (cf. \cite[Exercise XIII.2(c)]{McDonald74})} \label{regular-poly}
Let $R$ be a finite commutative local ring. Let $f(x) = a_0+a_1x+ \cdots +a_nx^n$ be in $R[x]$, then the following are equivalent.
\begin{enumerate}[{\rm (1)}]
 \item $f(x)$ is regular; \label{1}
\item  $\langle a_0, a_1, \cdots , a_n \rangle = R$; \label{4} 
\item  $a_i$ is an unit for some $i$, $0 \leq i \leq n$; \label{3}
\item  $\mu(f(x)) \neq 0$; \label{2}
\end{enumerate}
\end{prop}
\begin{proof}
 $(1) \Rightarrow (2):$ If possible, let $f(x)$ be a regular polynomial but $\langle a_0, a_1, \cdots ,$ $a_n \rangle$ $\neq R$. Since $\langle a_0, a_1, \cdots , a_n \rangle$ is an ideal and $R$ is a local ring, we have $\langle a_0, a_1, \cdots , a_n \rangle \subseteq M$. Therefore, all $a_i$ are non unit elements. In a finite local ring, if $M=\langle b_1, b_2, \cdots , b_n\rangle$ and the nilpotency index of $b_i$ are $t_i$ respectively, then the element $b_1^{t_1-1}b_2^{t_2-1}\cdots b_n^{t_n-1}$ is the zero divisor of all non unit elements. So the polynomial $g(x)=b_1^{t_1-1}b_2^{t_2-1}\cdots b_n^{t_n-1}$ is the zero divisor of $f(x)$. That is $f(x)g(x)=0$. Hence, a contradiction. Therefore, $\langle a_0, a_1, \cdots , a_n\rangle = R$.\\\\
$(2) \Rightarrow (3):$ If possible, let all $a_i$ are non unit. Then $a_i\in M$  for all $i$. This implies that $R=M$. This gives a contradiction. Hence, $a_i$ is an unit for some $i$, $0 \leq i \leq n$.\\\\
$(3)\Rightarrow(4):$ If $a_i$ is an unit, then $a_i \notin M.$ This gives $a_i+M \neq 0$. Hence, $\mu(f(x)) \neq 0$.\\\\
$(4)\Rightarrow(1):$ Let $\mu(f(x)) \neq 0$. We have $a_i+M \neq 0$ for some $i$. Therefore, $a_i$ is an unit in $R$. So there does not exists non zero element $d \in R$ such that $da_i=0$. By McCoy Theorem (See \cite[Theorem 2]{Mccoy42}), we can say that $f(x)$ is not a divisor zero. That is $f(x)$ is regular.
\end{proof}

The following version of the division algorithm holds true for polynomials over finite commutative local rings.
\begin{prop}{\rm (cf. \cite[Exercise XIII.6]{McDonald74})} \label{division-alg}
Let $R$ be a finite commutative local ring. Let $f(x)$ and $g(x)$ be non zero polynomials in $R[x]$. If $g(x)$ is regular, then there exist polynomials $q(x)$ and $r(x)$ in $R[x]$ such that $f(x) =g(x)q(x) + r(x)$ and $deg(r(x)) < deg(g(x))$.
\end{prop}
\subsection{{\rm The ring} $R_{u^2,v^2,p}$} $ $

Let $R_{u^2,v^2,p} = \Z_p + u\Z_p + v\Z_p + uv\Z_p, u^2=0$, $v^2=0$ and $uv = vu$. It is easy to see that the ring $R_{u^2,v^2,p}$ is a finite commutative local ring with unique maximal ideal $\langle u,v\rangle$. The set $\{ \{0\}, \langle u\rangle, \langle v\rangle, \langle u + \alpha v\rangle, \langle u,v\rangle\}$ gives all ideals of $R_{u^2, v^2, p}$, where $\alpha$ is a non zero element of $\Z_p$. Since the maximal ideal $\langle u,v\rangle$ is not principal, the ring $R_{u^2, v^2, p}$ is not a chain ring. 

Let $g(x)$ be a non zero polynomial in $\Z_p[x]$. By Proposition \ref{regular-poly}, it is easy to see that the polynomial $g(x)+up_1(x)+vp_2(x)+uvp_3(x) \in R_{u^2,v^2,p}[x]$ is regular. Note that $\text{deg}(g(x)+up_1(x)+vp_2(x)+uvp_3(x))= \text{deg}(g(x))$.

\subsection{\rm The Gray Map}\label{graymap} $ $

Let $w_L$ and $w_H$ and denote the Lee weight and Hamming weight respectively. Following \cite{Yil-Kar11}, we define the Lee weight as follows
\begin{equation*}
 w_L(a + ub + vc +uvd) = w_H(a+b+c+d, c+d, b+d, d) ~~ \forall~ a, b, c, d \in \Z_p.
\end{equation*}
Again as in \cite{Yil-Kar11}, we define the Gray map $\phi_L$ : $R_{u^2, v^2, p} \rightarrow \Z_{p}^4$ as follows
\begin{equation*}
 \phi_L(a + ub + vc +uvd) = (a+b+c+d, c+d, b+d, d) ~~ \forall~ a, b, c, d \in \Z_p.
\end{equation*}
The Gray map naturally extend to $R_{u^2, v^2, p}^n$ as distance preserving isometry $\phi_L$ : ($R_{u^2, v^2, p}^n, \text{Lee Weight}) \rightarrow (\Z_{p}^{4n}, \text{Hamming weight})$ as follows
\begin{equation*}
 \phi_L(a_1, a_2, \cdots, a_n) = (\phi_L(a_1), \phi_L(a_2), \cdots, \phi_L(a_n)) ~~ \forall~ a_i\in R_{u^2, v^2, p}.
\end{equation*}
By linearity of the map $\phi_L$ we obtain the following.
\begin{theo}
If $C$ is a linear code over  of length $n$, size $p^k$ and
minimum Lee weight $d$, then $\phi_L(C)$ is a p-ary linear code with parameters $[4n, k, d]$.

\end{theo}

\section{A generator for cyclic codes over the ring $R_{u^2, v^2, p}$}\label{generator}
Let $p$ be a prime number and $n$ be a positive integer. Let $R_{u^2,v^2,p} = \Z_p + u\Z_p + v\Z_p + uv\Z_p, u^2=0$, $v^2=0$ and $uv = vu$. We can write $R_{u^2,v^2,p}$ as $R_{u^2,v^2,p}= R_{u^2,p} + vR_{u^2,p}, v^2=0$, where $R_{u^2,p}=\Z_p + u\Z_p, u^2=0$. Let $R_{u^2,v^2,p,n}=R_{u^2,v^2,p}[x]/\langle x^n-1\rangle$. Let $C$ be a cyclic code of length $n$ over $R_{u^2,v^2,p}$. We also consider $C$ as an ideal in $R_{u^2,v^2,p,n}$. We define the map $\psi : R_{u^2,v^2,p} \rightarrow R_{u^2,p}$ by $\psi(\alpha + v \beta) = \alpha$, where $\alpha, \beta \in R_{u^2,p}$ . Clearly the map $\psi$ is a ring homomorphism. We extend this homomorphism to a homomorphism $\phi$ from $C$ to the ring $R_{u^2,p,n}$ defined by
\[\phi (c_0+c_1x+\cdots+c_{n-1}x^{n-1})=\psi(c_0)+\psi(c_1)x+\cdots+\psi(c_{n-1})x^{n-1},\]where $c_i \in R_{u^2,v^2,p}$. Let $J=\{r(x)\in R_{u^2,p,n}[x]: vr(x) \in \text{ker}\phi\}$. We see that $J$ is an ideal of $R_{u^2,p,n}$. Hence we can consider $J$ as a cyclic code over $R_{u^2,p}$. Now we know from \cite{Aks-Pkk13} that any ideal of $R_{u^2,p,n}$ is of the form $\langle h(x)+uq(x),ub(x)\rangle$, where $h(x), q(x), b(x) \in \Z_p[x]$, $b(x)|h(x)|(x^n-1)$ 
mod $p$ and $b(x)|q(x)\left(\frac{x^n-1}{h(x)}\right)$. So $J=\langle h(x)+uq(x),ub(x)\rangle$. Therefore, we can write $\text{ker}\phi=\langle vh(x)+uvq(x),uvb(x)\rangle$. Since $\phi$ is a homomorphism, the image $\text{Im}\phi$ is an ideal of $R_{u^2,p,n}$. Hence, $\text{Im}\phi$ is a cyclic code over $R_{u^2,p}$. Again we can write $\text{Im}\phi$ as above. That is $\text{Im}\phi=\langle g(x)+up(x),ua(x)\rangle$, where $a(x)|g(x)|(x^n-1)$ mod $p$ and $a(x)|p(x)\left(\frac{x^n-1}{g(x)}\right)$. So the code $C$ can be written as $C=\langle g(x)+up_1(x)+vq_1(x)+uvr_1(x),ua_1(x)+vq_2(x)+uvr_2(x),va_2(x)+uvr_3(x),uva_3(x)\rangle$, where $a_3(x)|a_1(x)|g(x)|(x^n-1)$ and $a_3(x)|a_2(x)|g(x)|(x^n-1)$.\\

For an ideal $C$ of the ring $R_{u^2,v^2,p,n}=R_{u^2,v^2,p}[x]/\langle x^n-1\rangle$, we define the residue and the torsion of the ideal $C$ as (see \cite{Dou-Yil-Kar12})
\begin{align*}
&\text{Res}(C)=\{a\in R_{u^2,p,n}|~\exists ~b\in R_{u^2,p,n}: a+vb\in C\} ~\text{and}\\
&\text{Tor}(C)=\{a\in R_{u^2,p,n}|~va\in C\}
\end{align*}
It can be easily shown that when $C$ is an ideal of $R_{u^2,v^2,p,n}$, $\text{Res}(C)$ and $\text{Tor}(C)$ both are ideals of $R_{u^2,p,n}$. And also it is easy to show that $\text{Res}(C)=\text{Im}\phi$ and $\text{Tor}(C)=J$. Now we define four ideals associated to C.
\begin{align*}
&C_1=\text{Res}(Res(C)) = C ~\text{mod} \langle u,v\rangle,\\
&C_2=\text{Tor}(\text{Res}(C)) = \{f(x)\in \Z_p[x]~|~uf(x)\in C ~\text{mod} ~v\},\\
&C_3=\text{Res}(\text{Tor}(C)) = \{f(x)\in \Z_p[x]~|~vf(x)\in C ~\text{mod}~ uv\} ~\text{and}\\
&C_4=\text{Tor}(\text{Tor}(C)) = \{f(x)\in \Z_p[x]~|~uvf(x)\in C\}.
\end{align*}
These are ideals of $\Z_p[x]/\langle x^n-1\rangle$, hence principal ideal and we have $C_1=\langle g(x)\rangle$, $C_2=\langle a_1(x)\rangle$, $C_3=\langle a_2(x)\rangle$, $C_4=\langle a_3(x)\rangle$.
\begin{theo}
Any ideal $C$ of the ring $R_{u^2,v^2,p,n}$ is uniquely generated by the polynomial $A_1(x)=g(x)+up_1(x)+vq_1(x)+uvr_1(x)$, $A_2(x)=ua_1(x)+vq_2(x)+uvr_2(x)$, $A_3(x)=va_2(x)+uvr_3(x)$ and $A_4(x)=uva_3(x)$ where $p_i(x)$'s, $q_i(x)$'s, $r_i(x)$'s are zero or ${\rm deg}(p_1(x))<{\rm deg}(a_1(x))$, ${\rm deg}(q_1(x))<{\rm deg}(a_2(x))$, ${\rm deg}(r_1(x))<{\rm deg}(a_3(x))$, ${\rm deg}(q_2(x))<{\rm deg}(a_2(x))$, ${\rm deg}(r_2(x))<{\rm deg}(a_3(x))$, ${\rm deg}(r_3(x))<{\rm deg}(a_3(x))$.
\end{theo}
\begin{proof} We prove the degree results for the $p_1(x)$, $q_1(x)$ and $r_1(x)$ only. Others follow similarly. Let $A_1(x)\neq 0$ and $\text{deg}(p_1(x))\geq \text{deg}(a_1(x))$. Then by division algorithm, we have $p_1(x)=a_1(x)q(x)+r(x)$ where $\text{deg}(r(x))<\text{deg}(a_1(x))$ or $r(x)=0$. Now $A_1(x)-q(x)A_2(x)=g(x)+ur(x)+v(q_1(x)-q(x)q_2(x))+uv(r_1(x)-q(x)r_2(x))\in C$. Now if $\text{deg}(q_1(x)-q(x)q_2(x))\geq \text{deg}(a_2(x))$ then by division algorithm, we have $q_1(x)-q(x)q_2(x)=a_2(x)q'(x)+r'(x)$, where $\text{deg}(r'(x))<\text{deg}(a_2(x))$. Therefore, $A_1(x)-q(x)A_2(x)-q'(x)A_3(x)=g(x)+ur(x)+vr'(x)+uv(r_1(x)-q(x)r_2(x)-q'(x)r_3(x))\in C$. Again by division algorithm, we have $r_1(x)-q(x)r_2(x)-q'(x)r_3(x)=a_3(x)q''(x)+r''(x)$, where $\text{deg}(r''(x))<\text{deg}(a_3(x))$ and hence $A_1(x)-q(x)A_2(x)-q'(x)A_3(x)-q''(x)A_4(x)=g(x)+ur(x)+vr'(x)+uvr''(x)\in C$. Now this polynomial satisfies our required property and also the polynomial $A_1(x)$ can be replaced by this polynomial. Now we have 
to prove that the polynomials $A_i(x)$'s are unique. Here again, we prove the uniqueness only for polynomial $A_1(x)$. If possible, let $A_1(x)=g(x)+up_1(x)+vq_1(x)+uvr_1(x)$ and $B_1(x)=g(x)+up'_1(x)+vq'_1(x)+uvr'_1(x)$ are two polynomial with same property in $C$. Hence, $A_1(x)-B_1(x)=u(p_1(x)-p'_1(x))+v(q_1(x)-q'_1(x))+uv(r_1(x)-r'_1(x))\in C$. Since $p_1(x)-p'_1(x)\in C_2$ and $\text{deg}(p_1(x)-p'_1(x))<\text{deg}(a_1(x))$ so $p_1(x)-p'_1(x)=0$ or $p_1(x)=p'_1(x)$. Similarly we have $q_1(x)=q'_1(x)$ and $r_1(x)=r'_1(x)$. Therefore, $A_1(x)=B_1(x)$.
\end{proof}

\begin{prop} \label{gen-main}
Let $C=\langle g(x)+up_1(x)+vq_1(x)+uvr_1(x),ua_1(x)+vq_2(x)+uvr_2(x),va_2(x)+uvr_3(x),uva_3(x)\rangle$ be an ideal of the ring $R_{u^2,v^2,p,n}$. Then we must have
\begin{align}
&a_3(x)|g(x), a_3(x)|a_1(x), a_3(x)|a_2(x), a_2(x)|g(x), a_1(x)|g(x), g(x)|(x^n-1),\\
&a_1(x)|p_1(x)\left(\frac{x^n-1}{g(x)}\right),\\
&a_2(x)|\frac{x^n-1}{g(x)}\left(q_1(x)-\frac{p_1(x)}{a_1(x)}q_2(x)\right),\\
&a_2(x)|\frac{g(x)}{a_1(x)}q_2(x),\\
&a_3(x)|q_2(x),\\
&a_3(x)|\frac{x^n-1}{a_2(x)}r_3(x),\\
&a_3(x)|\frac{x^n-1}{a_1(x)}\left(r_2(x)-\frac{q_2(x)}{a_2(x)}r_3(x)\right),\\
&a_3(x)|\left(p_1(x)-\frac{g(x)}{a_2(x)}r_3(x)\right),\\
&a_3(x)|\left(q_1(x)-\frac{g(x)}{a_1(x)}r_2(x)+\frac{g(x)}{a_1(x)a_2(x)}q_2(x)r_3(x)\right),
\end{align}
\begin{align}
&a_3(x)|\frac{x^n-1}{g(x)}\left(r_1(x)-\frac{p_1(x)}{a_1(x)}r_2(x)+\frac{q_1(x)+\frac{p_1(x)}{a_1(x)}q_2(x)}{a_2(x)}r_3(x)\right).
\end{align}
\end{prop}
\begin{proof} 
Conditions given in $(1)$ is clear from the above discussion.
Condition $(2)$ follows from the facts that $\frac{x^n-1}{g(x)}(g(x)+up_1(x)+vq_1(x)+uvr_1(x))=u\frac{x^n-1}{g(x)}p_1(x)+v\frac{x^n-1}{g(x)}q_1(x)+uv\frac{x^n-1}{g(x)}r_1(x)\in C$. That is $u\frac{x^n-1}{g(x)}p_1(x)\in C ~\text{mod}~ v$. Which implies that $\frac{x^n-1}{g(x)}p_1(x)\in C_2$. So $a_1(x)|\frac{x^n-1}{g(x)}p_1(x)$.
For condition $(3)$, we can write $\frac{x^n-1}{g(x)}(g(x)+up_1(x)+vq_1(x)+uvr_1(x))-\frac{x^n-1}{g(x)}\frac{p_1(x)}{a_1(x)}(ua_1(x)+vq_2(x)+uvr_2(x))=v\frac{x^n-1}{g(x)}\big(q_1(x)-\frac{p_1(x)}{a_1(x)}q_2(x)\big)+uv\frac{x^n-1}{g(x)}\big(r_1(x)-\frac{p_1(x)}{a_1(x)}r_2(x)\big)\in C$. Now $v\frac{x^n-1}{g(x)}\big(q_1(x)-\frac{p_1(x)}{a_1(x)}q_2(x)\big)\in C ~\text{mod}~ uv$. So $\frac{x^n-1}{g(x)}\big(q_1(x)-\frac{p_1(x)}{a_1(x)}q_2(x)\big)\in C_3$. Hence, $a_2(x)|\frac{x^n-1}{g(x)}\big(q_1(x)-\frac{p_1(x)}{a_1(x)}q_2(x)\big)$.
For condition $(4)$, we have $u(g(x)+up_1(x)+vq_1(x)+uvr_1(x))-\frac{g(x)}{a_1(x)}(ua_1(x)+vq_2(x)+uvr_2(x))=-v\big(\frac{g(x)}{a_1(x)}q_2(x)\big)+uv\big(q_1(x)-\frac{g(x)}{a_1(x)}r_2(x)\big)\in C$. Now $-v\big(\frac{g(x)}{a_1(x)}q_2(x)\big)\in C ~\text{mod}~uv$. This gives $-\big(\frac{g(x)}{a_1(x)}q_2(x)\big)\in C_3$. So we have $a_2(x)|\frac{g(x)}{a_1(x)}q_2(x)$.
For condition $(5)$, we have $u(ua_1(x)+vq_2(x)+uvr_2(x))=uvq_2(x)\in C$. Therefore, $q_2(x)\in C_4$. Hence, $a_3(x)|q_2(x)$.
For condition $(6)$, we have $\frac{x^n-1}{a_2(x)}(va_2(x)+uvr_3(x))=uv\frac{x^n-1}{a_2(x)}r_3(x)\in C$. So $\frac{x^n-1}{a_2(x)}r_3(x)\in C_4$. Hence, $a_3(x)|\frac{x^n-1}{a_2(x)}r_3(x)\in C$.
Condition $(7)$ follows from the fact that $\frac{x^n-1}{a_1(x)}(ua_1(x)+vq_2(x)+uvr_2(x))-\frac{x^n-1}{a_1(x)}\frac{q_2(x)}{a_2(x)}(va_2(x)+uvr_3(x))=uv\frac{x^n-1}{a_1(x)}\big(r_2(x)-\frac{q_2(x)}{a_2(x)}r_3(x)\big)\in C$. So $\frac{x^n-1}{a_1(x)}\big(r_2(x)-\frac{q_2(x)}{a_2(x)}r_3(x)\big)\in C_4$. Hence, $a_3(x)|\frac{x^n-1}{a_1(x)}\big(r_2(x)-\frac{q_2(x)}{a_2(x)}r_3(x)\big)$.
For condition $(8)$, we have $v(g(x)+up_1(x)+vq_1(x)+uvr_1(x))-\frac{g(x)}{a_2(x)}(va_2(x)+uvr_3(x))=uv\big(p_1(x)-\frac{g(x)}{a_2(x)}r_3(x)\big)\in C$. So $\big(p_1(x)-\frac{g(x)}{a_2(x)}r_3(x)\big)\in C_4$. Therefore, $a_3(x)|\big(p_1(x)-\frac{g(x)}{a_2(x)}r_3(x)\big)$.
For condition $(9)$, we can write $u(g(x)+up_1(x)+vq_1(x)+uvr_1(x))-\frac{g(x)}{a_1(x)}(ua_1(x)+vq_2(x)+uvr_2(x))+\frac{g(x)q_2(x)}{a_1(x)a_2(x)}(va_2(x)+uvr_3(x))=uv\big(q_1(x)-\frac{g(x)}{a_1(x)}r_2(x)+\frac{g(x)}{a_1(x)a_2(x)}q_2(x)r_3(x)\big)\in C$. This gives $\big(q_1(x)-\frac{g(x)}{a_1(x)}r_2(x)+\frac{g(x)}{a_1(x)a_2(x)}q_2(x)r_3(x)\big)\in C_4$. Hence, $a_3(x)|\big(q_1(x)-\frac{g(x)}{a_1(x)}r_2(x)+\frac{g(x)}{a_1(x)a_2(x)}q_2(x)r_3(x)\big)$.
Finally for condition $(10)$, we can write $\frac{x^n-1}{g(x)}(g(x)+up_1(x)+vq_1(x)+uvr_1(x))-\frac{x^n-1}{g(x)}\frac{p_1(x)}{a_1(x)}(ua_1(x)+vq_2(x)+uvr_2(x))+\frac{x^n-1}{g(x)}\frac{q_1(x)+\frac{p_1(x)}{a_1(x)}q_2(x)}{a_2(x)}(va_2(x)+uvr_3(x))=uv\frac{x^n-1}{g(x)}\big(r_1(x)-\frac{p_1(x)}{a_1(x)}r_2(x)+\frac{q_1(x)+\frac{p_1(x)}{a_1(x)}q_2(x)}{a_2(x)}r_3(x)\big)\in C$. This implies that $\frac{x^n-1}{g(x)}\big(r_1(x)-\frac{p_1(x)}{a_1(x)}r_2(x)+\frac{q_1(x)+\frac{p_1(x)}{a_1(x)}q_2(x)}{a_2(x)}r_3(x)\big) \in C_4$. Hence, we have $a_3(x)|\frac{x^n-1}{g(x)}\big(r_1(x)-\frac{p_1(x)}{a_1(x)}r_2(x)+\frac{q_1(x)+\frac{p_1(x)}{a_1(x)}q_2(x)}{a_2(x)}r_3(x)\big)$.
\end{proof}
The following theorem characterizes the free cyclic codes over  $R_{u^2,v^2,p}$.
\begin{prop} \label{freecode}
If $C=\langle g(x)+up_1(x)+vq_1(x)+uvr_1(x),ua_1(x)+vq_2(x)+uvr_2(x),va_2(x)+uvr_3(x),uva_3(x)\rangle$ be a cyclic code over the ring $R_{u^2,v^2,p}$, then $C$ is a free cyclic code if and only if $g(x)=a_3(x)$. In this case, we have $C=\langle g(x)+up_1(x)+vq_1(x)+uvr_1(x)\rangle$ and $(g(x)+up_1(x)+vq_1(x)+uvr_1(x))|(x^n-1)$ in $R_{u^2,v^2,p}$.
\end{prop}
\begin{proof} Let $g(x)=a_3(x)$. Since $a_3(x)|a_1(x)|g(x)$ and $a_3(x)|a_2(x)|g(x)$, we get $g(x)=a_1(x)=a_2(x)=a_3(x)$. Here, we have Im$\phi$ = $\langle g(x)+up_1(x), ua_1(x)\rangle$ and Ker$\phi$ = $ v\langle a_2(x)+ ur_3(x), ua_3(x) \rangle$.  From \cite[Lemma 3.1]{Aks-Pkk13}, we get Im$\phi$ = $\langle g(x)+up_1(x)\rangle$ and Ker$\phi$ = $ v\langle a_2(x)+ ur_3(x) \rangle$. Therefore, we have $C=\langle g(x)+up_1(x)+vq_1(x)+uvr_1(x), va_2(x)+uvr_3(x)\rangle$. Note that deg$(p_1(x))$, deg$(r_3(x))$ $<$ deg$(g(x))$. From Condition (8), we have $a_3(x)|(p_1(x)-r_3(x)),$ this gives $p_1(x) = r_3(x)$. So $ v(g(x)+up_1(x)+vq_1(x)+uvr_1(x)) = vg(x)+uvp_1(x) = va_2(x)+uvr_3(x).$ Therefore, we have $C=\langle g(x)+up_1(x)+vq_1(x)+uvr_1(x)\rangle$ and $C \simeq R_{u^2,v^2,p}^{n-{\rm deg}(g_1(x))}$  Hence, $C$ is a free cyclic code. Conversely, if $C$ is a free cyclic code, we must have $C =\langle g(x)+up_1(x)+vq_1(x)+uvr_1(x)\rangle$. Since $uv a_3(x) \in C$, we have $uva_3(x) = uv \alpha g(x)$ for some $\
\alpha \in \Z_p$. Note that $a_3(x)|g(x)$, hence by comparing the coefficients both sides, we get $g(x) = a_3(x)$. For the second condition, by division algorithm, we have $x^n-1=(g(x)+up_1(x)+vq_1(x)+uvr_1(x))q(x)+r(x)$, where $r(x)=0$ or $\text{deg}(r(x))<\text{deg}(g(x))$. This implies that $r(x)=(x^n-1)-(g(x)+up_1(x)+vq_1(x)+uvr_1(x))q(x)\in C$. Note that $g(x)+up_1(x)+vq_1(x)+uvr_1(x)$ is the lowest degree polynomial in $C$. So $r(x)=0$. Hence, $(g(x)+up_1(x)+vq_1(x)+uvr_1(x))|(x^n-1)$ in $R_{u^2,v^2,p}$.
\end{proof}
Note that we get the simpler form for the generators of the cyclic code over $R_{u^2,v^2,p}$, like in the above theorem, if we have $g(x) = a_1(x), a_2(x)$ or $a_3(x) = a_1(x), a_2(x)$.
\subsection{\rm If $n$ is relatively prime to $p$} $ $

Let $n$ be relatively prime to $p$. If $C$ is a cyclic code of length $n$ over the ring $R_{u^2,v^2,p}$, then from \cite{Aks-Pkk13}, we have Im$(\phi)$ = $\langle g(x) + ua_1(x) \rangle $and ker$\phi$ = $ v\langle a_2(x)+ ua_3(x) \rangle$
with $a_1(x)|g(x)|(x^n-1)$ and $ a_3(x)|a_2(x)|(x^n-1).$  Now combining these two we can write
$C  = \langle g_1(x) + ua_1 (x) + vq_1(x)+uvr_1(x), va_2 (x) + uva_3(x) \rangle$
with the same conditions as above on $g(x)$ and $a_i(x)$.
In order to get conditions on $q_1(x)$ and $r_1(x)$, we will define a homomorphism $\phi_v : R_{u^2,v^2,p} \rightarrow R_{v^2,p}$ as 
\[ \phi_v(a+bu+cv+duv) = a+cv.\]
Now with $C$ as above, we get $\phi_v(C) = \langle g_1(x) + vq_1(x), va_2 (x) \rangle.$ Note that $\phi_v(C)$ is a cyclic code of length $n$ over $R_{v^2,p}$. Therefore, from \cite{Aks-Pkk13}, we have $a_2(x)|g_1(x)|(x^n-1)$, $a_2(x)|q_1(x)\dfrac{x^n-1}{a_1(x)}$ and deg$(q_1(x)) <$ deg$(a_2(x))$. Since $n$ is relatively prime to $p$, $x^n-1$ can be uniquely factored as product of distinct irreducible factors. Therefore, we must have g.c.d.$\left(a_2(x), \dfrac{x^n-1}{a_1(x)}\right) = 1$. But we have $a_2(x)|q_1(x)\dfrac{x^n-1}{a_1(x)}$, this gives $a_2(x)|q_1(x)$. Since deg$(q_1(x)) <$ deg$(a_2(x))$, this gives $q_1(x) = 0$. Thus we have proved the following theorem.
\begin{theo}
Let $C$ be a cyclic code over the ring $R_{u^2,v^2,p}$ of length $n$. If $n$ is relatively prime to $p$, then we have $C=\langle g(x)+ua_1(x)+uvr_1(x), va_2(x)+uva_3(x)\rangle$ with $a_1(x)|g(x)|(x^n-1)$ and $ a_3(x)|a_2(x)|g(x)|(x^n-1).$
\end{theo}
%Note that if $g(x) = a_2(x)$, then as in Proposition \ref{freecode}, we get$ C=\langle g(x)+up_1(x)+uvr_1(x), va_2(x)+uvp_2(x)\rangle$.

\section{Ranks and minimal spanning sets}\label{rank}
In this section, we find the rank and minimal spanning set of cyclic codes over the ring $R_{u^2,v^2,p}$. Following Dougherty and Shiromoto \cite[page 401]{Dou-Shiro01}, we define the rank of the code $C$ by the minimum number of generators of $C$ and define the free rank of $C$ by the maximum of the ranks
of $R_{u^2,v^2,p}$-free submodules of $C$.

\begin{theo} \label{rank-main}
Let $n$ be a positive integer not relatively prime to $p$ and $C$ be a cyclic code of length $n$ over $R_{u^2,v^2,p}$. If $C=\langle g(x)+up_1(x)+vq_1(x)+uvr_1(x),ua_1(x)+vq_2(x)+uvr_2(x),va_2(x)+uvr_3(x),uva_3(x)\rangle$ with $t ={\rm deg}(g(x))$, $ t_1 = {\rm deg}(a_1(x))$, $t_2 = {\rm deg}(a_2(x))$, $ t' = {\rm min}\{{\rm deg}(a_1(x)),{\rm deg}(a_2(x))\}$ and $t_3 = {\rm deg}(a_3(x))$, then $C$ has free rank $n-t$, rank $n+t+t'-t_1-t_2-t_3$ and the minimal spanning set  $B=\{g(x)+up_1(x)+vq_1(x)+uvr_1(x), x(g(x)+up_1(x)+vq_1(x)+uvr_1(x)),$ $\cdots, x^{n-t-1}(g(x)+up_1(x)+vq_1(x)+uvr_1(x)), ua_1(x)+vq_2(x)+uvr_2(x), x(ua_1(x)+vq_2(x)+uvr_2(x)),\cdots, x^{t-t_1-1}$ $(ua_1(x)+vq_2(x)+uvr_2(x)), $ $va_2(x)+uvr_3(x), x(va_2(x)+uvr_3(x)), $ $\cdots, x^{t-t_2-1}$ $(va_2(x)+uvr_3(x)), uva_3(x),\\ x(uva_3(x)),\cdots,  x^{t'-t_3-1}(uva_3(x))\}$. Furthermore, if $q_2(x) \neq 0$, we have $|C| = p^{4n+t+t'-3t_1-2t_2-t_3}$ and if $ q_2(x) = 0$, we have $|C| = p^{4n+t'-2t_1-2t_2-t_3}$.
\end{theo}
\begin{proof}
Let $ t' = \text{deg}a_2(x).$  It is suffices to show that $B$ spans the set $B'=\{g(x)+up_1(x)+vq_1(x)+uvr_1(x), x(g(x)+up_1(x)+vq_1(x)+uvr_1(x)),$ $ \cdots , x^{n-t-1}\\(g(x)+up_1(x)+vq_1(x)+uvr_1(x)), ua_1(x)+vq_2(x)+uvr_2(x), x(ua_1(x)+vq_2(x)+uvr_2(x)), $ $\cdots , x^{n-t_1-1}(ua_1(x)+vq_2(x)+uvr_2(x)), va_2(x)+uvr_3(x), $ $x(va_2(x)+uvr_3(x)), \cdots , x^{n-t_2-1}(va_2(x)+uvr_3(x)), uva_3(x), x(uva_3(x)), \cdots, $ $ x^{n-t_3-1}(uva_3(x))\}$. First we show that $uvx^{t_2-t_3}(a_3(x)) \in B$. Let the leading coefficient of $x^{t_2-t_3}(a_3(x))$ be $\alpha_3$ and of $a_2(x)+ur_3(x)$ be $\alpha_2$. Then there exist a constant $c_0\in \Z_p$ such that $\alpha_3=c_0\alpha_2$. By the division algorithm, we have $uvx^{t_2-t_3}(a_3(x))=uvc_0(a_2(x)+ur_3(x))+uvm(x)$, where $uvm(x)$ is a polynomial in $C$ such that $\text{deg}(m(x))<t_2$. Now $\text{deg}(m(x))\geq t_3$ because $a_3(x)$ is minimum degree polynomial such that $uva_3(x)\in C$. So $t_3\leq \text{deg}(m(x)<t_2$ and $uvm(x)=\alpha_0uva_3(x)+\alpha_1uva_
3(x)+\cdots +\alpha_{t_2-t_3-1}uva_3(x)$. Thus $uvx^{t_2-t_3}(a_3(x))\in B$. Let $ t' =\text{deg}a_1(x).$  Since $(ua_1(x)+vq_2(x)+uvr_2(x))$ is not regular, we can not take $ua_1(x)+vq_2(x)+uvr_2(x)$ as divisor in the division algorithm. Here, by direct computation, we show that $uvx^{t_1-t_3}(a_3(x)) \in B$. Since $a_3(x)|a_1(x)$, we have $a_1(x)=q(x)a_3(x)$ for some $q(x) \in \Z_p[x]$. This we can write as $ a_1(x) = a_3(x)(q_0+q_1x +\cdots +q_{t_1-t_3}x^{t_1-t_3})$. Clearly, $q_{t_1-t_3} \neq 0$ (by degree comparison). Thus, we have $ uva_3(x)x^{t_1-t_3} $ = $ uvq_{t_1-t_3}^{-1}a_1(x)-uva_3(x)(q_0+q_1x +\cdots +q_{t_1-t_3-1}x^{t_1-t_3-1})$. Note that $v(ua_1(x)+vq_2(x)+uvr_2(x))= uva_1(x),$ thus, $uva_3x^{t_1-t_3} \in B$. By taking $g(x)+up_1(x)+vq_1(x)+uvr_1(x)$ as divisor and applying the division algorithm for $ua_1(x)+vq_2(x)+uvr_2(x)$ and $ va_2(x)+uvr_3(x)$, we can show that $x^{t-t_1}(ua_1(x)+vq_2(x)+uvr_2(x)), x^{t-t_2}(va_2(x)+uvr_3(x)) \in B $. Similarly others also can be shown. So $B$ is a 
generating set. It is easy to see that the elements $x^{i-t-1}(g(x)+up_1(x)+vq_1(x)+uvr_1(x)), 1 \leq i \leq n,~ x^{j-t_1-1}(ua_1(x)+vq_2(x)+uvr_2(x)), 1 \leq j \leq t,~ x^{k-t_2-1}(va_2(x)+uvr_3(x)), ~1 \leq k \leq t$ and $x^{l-t_3-1}uv(a_3(x)), 1 \leq l \leq t'$ can not be written as the linear combination of its preceding elements and other elements in the set $B$.
Here we will only show that $ x^{t-t_1-1}(ua_1(x)+vq_2(x)+uvr_2(x))$ can not be written as linear combinations of $g(x)+up_1(x)+vq_1(x)+uvr_1(x), x(g(x)+up_1(x)+vq_1(x)+uvr_1(x)),$ $\cdots, x^{n-t-1}(g(x)+up_1(x)+vq_1(x)+uvr_1(x)), ua_1(x)+vq_2(x)+uvr_2(x), x(ua_1(x)+vq_2(x)+uvr_2(x)),\cdots,$ $ x^{t-t_1-2}(ua_1(x)+vq_2(x)+uvr_2(x)), $ $va_2(x)+uvr_3(x), x(va_2(x)+uvr_3(x)), $ $\cdots, x^{t-t_2-1}(va_2(x)+uvr_3(x)), uva_3(x),x(uva_3(x)),$ $\cdots, x^{t'-t_3-1}(uva_3(x))$. The proof is similar for the rest. Suppose $x^{t-t_1-1}(ua_1(x)+vq_2(x)+uvr_2(x))$ can be written as linear combinations of the above elements. Then we have 
$x^{t-t_1-1}(ua_1(x)+vq_2(x)+vur_2(x))=\alpha_1(g(x)+up_1(x)+vq_1(x)+vur_1(x))+\alpha_2x(g(x)+up_1(x)+vq_1(x)+vur_1(x))+\cdots+\alpha_{n-t}x^{n-t-1}(g(x)+up_1(x)+vq_1(x)+vur_1(x))+\beta_1(ua_1(x)+vq_2(x)+vur_2(x))+\beta_2x(ua_1(x)+vq_2(x)+vur_2(x))+\cdots+\beta_{t-t_1-2}x^{t-t_1-2}(ua_1(x)+vq_2(x)+vur_2(x))+\gamma_1(va_2(x)+vur_3(x))+\gamma_2x(va_2(x)+vur_3(x))+\cdots+\gamma_{t-t_2}x^{t-t_2-1}(va_2(x)+vur_3(x))+\delta_1x(vua_3(x))+\delta_2x(vua_3(x))+\cdots+\delta_{t'-t_3}x^{t'-t_3-1}(vua_3(x))$, where $\alpha_i=\alpha_{i1}+\alpha_{i2}u+\alpha_{i3}v+\alpha_{i4}uv$, $\beta_i=\beta_{i1}+\beta_{i2}u+\beta_{i3}v+\beta_{i4}uv$, $\gamma_i=\gamma_{i1}+\gamma_{i2}u+\gamma_{i3}v+\gamma_{i4}uv$ and $\delta_{i}=\delta_{i1}+\delta_{i2}u+\delta_{i3}v+\delta_{i4}uv$. We have
\begin{align*}
x^{t-t_1-1}(ua_1(x)+vq_2(x)+vur_2&(x))=(\alpha_{11}+\alpha_{21}x+\cdots+\alpha_{(n-t)1}x^{n-t-1})g(x)\\
&+u(\alpha_{12}+\alpha_{22}x+\cdots+\alpha_{(n-t)2}x^{n-t-1})g(x)\\
&+u(\alpha_{11}+\alpha_{21}x+\cdots+\alpha_{(n-t)1}x^{n-t-1})p_1(x)\\
&+u(\beta_{11}+\beta_{21}x+\cdots+\beta_{(t-t_1-1)1}x^{t-t_1-2})a_1(x)\\
%&+v(\alpha_{13}+\alpha_{23}x+\cdots+\alpha_{(n-t)3}x^{n-t-1})g(x)\\
%&+v(\alpha_{11}+\alpha_{21}x+\cdots+\alpha_{(n-t)1}x^{n-t-1})q_1(x)\\
%&+v(\beta_{11}+\beta_{21}x+\cdots+\beta_{(t-t_1-1),1}x^{t-t_1-2})q_2(x)\\
%&+v(\gamma_{11}+\gamma_{21}x+\cdots+\gamma_{(t-t_2)1}x^{t-t_2-1})a_2(x)\\
&+vm_1(x)
+uvm_2(x), 
\end{align*}
where $m_1(x)$ and $m_2(x)$ are polynomials in $\Z_p[x].$ By comparing both sides, we have $\alpha_{i1} = 0$, $ 1 \leq i \leq n-t$ and $x^{t-t_1-1}a_1(x) = (\alpha_{12}+\alpha_{22}x+\cdots+\alpha_{(n-t)2}x^{n-t-1})g(x) + (\beta_{11}+\beta_{21}x+\cdots+\beta_{(t-t_1-1)1}x^{t-t_1-2})a_1(x)$.
 Note that $ \text{deg}(x^{t-t_1-1}a_1(x)) = t-1$ but $\text{deg}((\alpha_{12}+\alpha_{22}x+\cdots+\alpha_{(n-t)2}x^{n-t-1})g(x)) \geq t$ and $\text{deg}((\beta_{11}+\beta_{21}x+\cdots+\beta_{(t-t_1-1)1}x^{t-t_1-2})a_1(x)) \leq t-2$. Hence, this gives a contradiction.
%Now we only need to show that $B$ is linearly independent. Suppose $\alpha_0(g(x)+up_1(x)+vq_1(x)+uvr_1(x))+ \alpha_1x(g(x)+up_1(x)+vq_1(x)+uvr_1(x))+ \cdots + \alpha_{n-t-1}x^{n-t-1}$ $(g(x)+up_1(x)+vq_1(x)+uvr_1(x))+ \beta_0(ua_1(x)+uvr_2(x))+ \beta_1x(ua_1(x)+uvr_2(x))+ \cdots + \beta_{t-t_1-1}x^{t-t_1-1}(ua_1(x)+uvr_2(x))+ \gamma_0(va_2(x)+uvr_3(x))+ \gamma_1x(va_2(x)+uvr_3(x))+ \cdots + \gamma_{t_1-t_2-1}x^{t_1-t_2-1}(va_2(x)+uvr_3(x))+ \delta_0(uva_3(x))+ \delta_1x(uva_3(x))+ \cdots + \delta_{t_2-t_3-1}x^{t_2-t_3-1}(uva_3(x))=0$. This implies that $(g(x)+up_1(x)+vq_1(x)+uvr_1(x))(\alpha_0+ \alpha_1x+ \cdots + \alpha_{n-t-1}x^{n-t-1})+ u(a_1(x)+vr_2(x))(\beta_0+ \beta_1x+ \cdots + \beta_{t-t_1-1}x^{t-t_1-1})+ v(
%a_2(x)
%+ut_3(
%x))(\gamma_0+ \gamma_1x+ \cdots + \gamma_{t_1-t_2-1}x^{t_1-t_2-1})+ uv(a_3(x))(\delta_0+ \delta_1x+ \cdots + \delta_{t_2-t_3-1}x^{t_2-t_3-1})=0$. This implies that $(g(x)+up_1(x)+vq_1(x)+uvr_1(x))(\alpha_0+ \alpha_1x+ \cdots + \alpha_{n-t-1}x^{n-t-1})=0$, $u(a_1(x)+vr_2(x))(\beta_0+ \beta_1x+ \cdots + \beta_{t-t_1-1}x^{t-t_1-1})=0$, $v(a_2(x)+ur_3(x))(\gamma_0+ \gamma_1x+ \cdots + \gamma_{t_1-t_2-1}x^{t_1-t_2-1})=0$ and $uv(a_3(x))(\delta_0+ \delta_1x+ \cdots + \delta_{t_2-t_3-1}x^{t_2-t_3-1})=0$. Since $g(x)+up_1(x)+vq_1(x)+uvr_1(x), a_1(x)+vr_2(x), a_2(x)+ur_3(x)$ and $a_3(x)$ are not zero divisors, this givies $ \alpha_i= \beta_j = \delta_k = \gamma_l = 0$ for all $ 1 \leq i \leq n-t-1, 1 \leq j \leq t-t_1-1, 1 \leq k \leq t_1-t_2-1$ and $ 1\leq l \leq t_2-t_3-1.$ Therefore, the set $B$ is linearly independent.
\end{proof}

%\begin{theo}
%Let $n$ be a positive integer not relatively prime to $p$ and $C$ be a cyclic code of length $n$ over $R_{u^2,v^2,p}$. If $C=\langle g(x)+up_1(x)+vq_1(x)+uvr_1(x),ua_1(x)+vq_2(x)+uvr_2(x),va_2(x)+uvr_3(x),uva_3(x)\rangle$ with $t ={\rm deg}(g(x))$, $ t_1 = {\rm deg}(a_1(x))$, $ t_2 = {\rm deg}(a_2(x))$, $t_3 = {\rm deg}(a_3(x))$ and $q_2(x)\neq 0$, then $C$ has rank $n+t-t_1-t_3$,  $|C| = p^{4n+t-3t_1-t_2-t_3}$. The minimal spanning set of $C$ is $B=\{g(x)+up_1(x)+vq_1(x)+uvr_1(x), x(g(x)+up_1(x)+vq_1(x)+uvr_1(x)),$ $\cdots, x^{n-t-1}(g(x)+up_1(x)+vq_1(x)+uvr_1(x)), ua_1(x)+vq_2(x)+uvr_2(x), x(ua_1(x)+vq_2(x)+uvr_2(x)),\cdots,$ $ x^{t-t_1-1}(ua_1(x)+vq_2(x)+uvr_2(x)), $ $va_2(x)+uvr_3(x), x(va_2(x)+uvr_3(x)), $ $\cdots, x^{t-t_2-1}(va_2(x)+uvr_3(x)), uva_3(x),\\ x(uva_3(x)),$ $\cdots , x^{t_2-t_3-1}(uva_3(x))\}$.
%\end{theo}
%\begin{proof}
%The proof is same as Theorem \ref{rank-main}. 
%\end{proof}

\begin{theo}
Let $n$ be a positive integer relatively prime to $p$ and $C$ be a cyclic code of length $n$ over $R_{u^2,v^2,p}$. If $C=\langle g(x)+ua_1(x)+uvr_1(x), va_2(x)+uva_3(x)\rangle$ with $t ={\rm deg}(g(x))$, $t_2 = {\rm deg}(a_2(x))$, then $C$ has rank $n-t_2$ and $|C| = p^{4n-2t-2t_2}$. The minimal spanning set of $C$ is  $B=\{g(x)+ua_1(x)+uvr_1(x), x(g(x)+ua_1(x)+uvr_1(x)),$ $\cdots, x^{n-t-1}(g(x)+ua_1(x)+uvr_1(x)),$ $va_2(x)+uva_3(x), x(a_2(x)+uva_3(x)), $ $\cdots, x^{t-t_2-1}(va_2(x)+uva_3(x))$.
\end{theo}
\begin{proof}
The proof is same as Theorem \ref{rank-main}.
\end{proof}
\section{Minimum distance} \label{md}
Let $n$ be a positive integer not relatively prime to $p$. Let $ C $ be a cyclic code of length $n$ over $R_{u^2,v^2,p}$. We define $ C_{uv} = \{k(x) \in R_{u^2,v^2,p,n} : uv k(x) \in C\}.$ It is easy to see that $C_{uv}$ is a cyclic code over $\Z_p$.
\begin{theorem} \label{md1}
Let $n$ be not relatively prime to $p$. If $ C = \langle g(x)+up_1(x)+vq_1(x)+uvr_1(x),ua_1(x)+vq_2(x)+uvr_2(x),va_2(x)+uvr_3(x),uva_3(x)\rangle$ is a cyclic code of length $n$ over $R_{u^2,v^2,p}$ Then $ C_{uv} = \langle a_{3}(x)\rangle$ and $w_{H}(C) = w_{H}(C_{uv}).$ \end{theorem}
\begin{proof}
We have $uv a_{3}(x) \in C$, thus $ \langle a_{3}(x)\rangle \subseteq C_{uv}.$ If $b(x) \in C_{uv}$, then $uvb(x) \in C$ and hence there exist polynomials $ b_1(x),b_2(x), b_3(x), b_4(x) \in \Z_p[x]$ such that $ uvb(x) = b_1(x)uvg(x) + b_2(x)uva_1(x) + b_3(x)uva_2(x) + b_4(x)uva_{3}(x).$
Since $a_{3}(x)|a_{2}(x)|g(x)$ and $a_{3}(x)|a_{1}(x)|g(x)$, we have $uvb(x) = m(x)uva_{3}(x)$ for some polynomial $m(x) \in \Z_p[x]$. So, $C_{uv} \subseteq \langle a_{3}(x)\rangle $, and hence $C_{uv} = \langle a_{3}(x)\rangle.$ Let $m(x) = m_0(x) + um_1(x) + v m_{2}(x) +uv m_3(x) \in C,$ where $m_0(x), m_1(x), m_2(x),\\ m_{3}(x) \in \Z_p[x].$
We have $uvm(x) = uvm_0(x),$ $w_{H}(uvm(x)) \leq w_{H}(m(x))$ and $uv C$ is subcode of $C$ with $w_{H}(uv C) \leq w_{H}(C)$. Therefore, it is sufficient to focus on the subcode $uv C$ in order to prove the theorem. Since $uv C = \langle uv a_{3}(x)\rangle$, we get $w_{H}(C) = w_{H}(C_{uv}).$
\end{proof}
\begin{definition}
Let $ m = b_{l-1}p^{l-1} + b_{l-2}p^{l-2} + \cdots + b_1p + b_0$, $b_i \in \Z_p, 0 
\leq i \leq l-1$, be the $p$-adic expansion of $m$.
\begin{enumerate} [{\rm (1)}]
 \item If $ b_{l-i}  \neq 0$ for all $1  \leq i \leq q, q < l, $ and $ b_{l-i} = 0 $ for all $i, q+1 \leq i \leq l$, then $m$ is said to have a $p$-adic length $q$ zero expansion.
\item If $ b_{l-i}  \neq 0$ for all $1  \leq i \leq q, q < l, $ $b_{l-q-1} = 0$ and $ b_{l-i} \neq 0 $ for some $i, q+2 \leq i \leq l$, then $m$ is said to have  $p$-adic length $q$ non-zero expansion.
\item If $ b_{l-i}  \neq 0$ for $1  \leq i \leq l, $ then $m$ is said to have a $p$-adic length $l$  expansion or $p$-adic full expansion.
\end{enumerate}
\end{definition}
\begin{lemma} \label{md-lemma}
Let $C$ be a cyclic code over $R_{u^2,v^2,p}$ of length $p^l$ where $l$ is a positive integer. Let $C = \langle a(x)\rangle$ where $a(x) = (x^{p^{l-1}} - 1)^bh(x)$, $ 1 \leq b < p$. If $h(x)$ generates a cyclic code of length $p^{l-1}$ and minimum distance $d$ then the minimum distance $d(C)$ of $C$ is $(b+1)d$.
\end{lemma}
\begin{proof}
For $ c \in C$, we have $ c = (x^{p^{l-1}} - 1)^bh(x)m(x)$ for some $ m(x) \in \frac{R_{u^2,v^2,p}[x]}{(x^{p^l}-1)}$. Since $h(x)$ generates a cyclic code of length $p^{l-1}$, we have $w(c) = w((x^{p^{l-1}} - 1)^bh(x)m(x)) = w(x^{p^{l-1}b}h(x)m(x)) + w(^bC_1x^{p^{l-1}(b-1)}h(x)m(x)) + \cdots + w(^bC_{b-1}\\x^{p^{l-1}}h(x)m(x)) + w(h(x)m(x))$. Thus, $ d(C) = (b + 1)d$.
\end{proof}

\begin{theorem} \label{md-thm}
Let $C$ be a cyclic code over $R_{u^2,v^2,p}$ of length $p^l$ where $l$ is a positive integer. Then,  $ C = \langle g(x)+up_1(x)+vq_1(x)+uvr_1(x),ua_1(x)+vq_2(x)+uvr_2(x),va_2(x)+uvr_3(x),uva_3(x)\rangle $ where $g(x) = (x-1)^{t}, a_1(x) = (x-1)^{t_1}, a_{2}(x) = (x-1)^{t_2}$ and $a_3(x) = (x-1)^{t_3}$ for some $t > t_1 > t_3 > 0$ and $t > t_2 > t_3 > 0$ 
\begin{enumerate}[{\rm (1)}]
\item If $t_3 \leq p^{l-1},$ then $d(C) = 2$. 
\item If $t_3 > p^{l-1},$ let $t_3 = b_{l-1}p^{l-1} + b_{l-2}p^{l-2} + \cdots + b_1p + b_0$ be the $p$-adic expansion of $t_3$ and $ a_{3}(x) = (x-1)^{t_3} = (x^{p^{l-1}} - 1)^{b_{l-1}}(x^{p^{l-2}} - 1)^{b_{l-2}} \cdots (x^{p^{1}} - 1)^{b_1}(x^{p^0} - 1)^{b_0}$.
\begin{enumerate}[{\rm ($a$)}]
 \item If $t_3$ has a $p$-adic length $q$ zero expansion or full expansion $(l=q)$, then $d(C) = (b_{l-1}+1)(b_{l-2}+1)\cdots(b_{l-q}+1).$
\item If $t_3$ has a $p$-adic length $q$ non-zero expansion, then $d(C) = 2(b_{l-1}+1)(b_{l-2}+1)\cdots(b_{l-q}+1).$
\end{enumerate}
\end{enumerate}
\end{theorem}
\begin{proof}
The first claim easily follows from Proposition \ref{gen-main}. From Theorem \ref{md1}, we see that $d(C) = d(uvC) = d(\langle(x - 1)^{t_3}\rangle)$. Hence, we only need to determine the minimum weight of $uvC = \langle(x - 1)^{t_3}\rangle.$\\
(1) If $t_3 \leq p^{l-1},$ then $(x - 1)^{t_3}(x - 1)^{p^{l-1}-t_3} = (x - 1)^{p^{l-1}} = (x^{p^{l-1}} - 1) \in C$. Thus, $d(C) = 2.$\\
(2) Let $ t_3 > p^{l-1}$. (a) If $t_3$ has a $p$-adic length $q$ zero expansion, we have $t_3 = b_{l-1} p^{l-1} + b_{l-2}p^{l-2} + \cdots + b_{l-q}p^{l-q}$, and $a_{3}(x) = (x - 1)^{t_3} = (x^{p^{l-1}}-1)^{b_{l-1}}(x^{p^{l-2}}-1)^{b_{l-2}}\cdots(x^{p^{l-q}}-1)^{b_{l-q}}.$ Let $ h(x) = (x^{p^{l-q}}-1)^{b_{l-q}}.$ Then $h(x)$ generates a cyclic code of length $ p^{l-q+1}$ and minimum distance $ (b_{l-q}+1)$. By Lemma \ref{md-lemma}, the subcode generated by $(x^{p^{l-q+1}}-1)^{b_{l-q+1}}h(x)$ has minimum distance $ (b_{l-q+1}+1) (b_{l-q}+1).$ By induction on $q$, we can see that the code generated by $a_{3}(x)$ has minimum distance $(b_{l-1}+1)(b_{l-2}+1)\cdots(b_{l-q}+1).$ Thus, $d(C) = (b_{l-1}+1)(b_{l-2}+1)\cdots(b_{l-q}+1).$\\
(b) If $t_3$ has a $p$-adic length $q$ non-zero expansion, we have $t_3 = b_{l-1} p^{l-1} + b_{l-2}p^{l-2} + \cdots + b_{1}p + b_0, b_{l-q-1} = 0.$ Let $ r = b_{l-q-2}p^{l-q-2}+ b_{l-q-3}p^{l-q-3}+ \cdots + b_1p + b_0$ and $ h(x) = (x-1)^r = (x^{p^{l-q-2}}-1)^{b_{l-q-2}}(x^{p^{l-q-3}}-1)^{b_{l-q-3}}\cdots(x^{p^{1}}-1)^{b_{1}}(x^{p^{0}}-1)^{b_{0}}.$ Since $ r < p^{l-q-1}$, we have $ p^{l-q-1} = r+j$ for some non-zero $j$. Thus, $ (x-1)^{p^{l-q-1}-j}h(x) = (x^{p^{l-q-1}} - 1) \in C.$ Hence, the subcode generated by $h(x)$ has minimum distance 2. By Lemma \ref{md-lemma}, the subcode generated by $ (x^{p^{l-q}} - 1)^{b_{l-q}}h(x)$ has minimum distance $2(b_{l-q}+1)$. By induction on $q$, we can see that the code generated by $a_{3}(x)$ has minimum distance $2(b_{l-1}+1)(b_{l-2}+1)\cdots(b_{l-q}+1).$ Thus, $d(C) = 2(b_{l-1}+1)(b_{l-2}+1)\cdots(b_{l-q}+1).$\\
\end{proof}

\section{Examples} \label{exm}
\begin{example}
Cyclic codes of length $3$ over $R_{u^2,v^2,3} = \Z_3 + u \Z_3 + v \Z_3 + uv \Z_3, u^2 = 0, v^2 = 0, uv = vu$: We have
$$ x^3-1 = (x-1)^3 ~\text{over}~ R_{u^2,v^2,3}.$$
Let $g = x-1$. The non-zero cyclic codes of length $3$ over $R_{u^2,v^2,3}$ with generator polynomial, rank and minimum distance are given in table below.
\end{example}

{\bf Table 1.} Non zero cyclic codes of length 3 over $R_{u^2,v^2,3}$.\\
\begin{tabular}{| l | c| c |}
\hline
Non-zero generator polynomials & Rank & d(C)\\
\hline
$<uvg^2>$ & 1 & 3\\
\hline
$<uvg>$ & 2 & 2\\
\hline
$<uv>$ & 3 & 1\\
\hline
$<vg^2+uvc_0g>$ & 1 & 3\\
\hline
$<vg^2+uvc_0, uvg>$ & 2 & 2\\
\hline
$<vg^2, uv>$ & 3 & 1\\
\hline
$<vg+uvc_0>$ & 2 & 2\\
\hline
$<vg, uv>$ & 3 & 1\\
\hline
$< v >$ & 3 & 1\\
\hline
$< ug^2+vc_0g^2+uvc_1g >$& 1 & 3\\
\hline
$< ug^2+vc_0g^2+uvc_1, uvg>$& 2&2\\
\hline
$< ug^2+vc_0g^2, uv>$&3&1\\
\hline
$< ug^2+uvc_0g, vg^2+uvc_1g >$&2&3\\
\hline
$< ug^2+vc_0g+uvc_1, vg^2+uvc_2, uvg>$, $c_0c_2=0$&3&2\\
\hline

\end{tabular}

\begin{tabular}{| l |c |c|}

\hline
Non-zero generator polynomials & Rank & d(C)\\
\hline
$< ug^2+vc_0g, vg^2, uv>$&4&1\\
\hline
$< ug^2+uvc_0, vg+uvc_1 >$&3&2\\
\hline
$<ug^2+vc_0, vg, uv>$&4&1\\
\hline
$< ug^2, v >$&4&1\\
\hline
$< ug+vc_0g+uvc_1>$&2&2\\
\hline
$< ug+vc_0g, uv>$&3&1\\
\hline
$< ug+vc_0g+uvc_1, vg^2+uvc_2>$&3&2\\
\hline
$< ug+v(c_0x+c_1), vg^2, uv>$&4&1\\
\hline
$< ug+uvc_0, vg+uvc_1 >$&4&2\\
\hline
$< ug+vc_0, vg, uv>$&5&1\\
\hline
$< ug, v >$&5&1\\
\hline
$< u+v(c_0x^2+c_1x+c_2) >$&3&1\\
\hline
$< u+v(c_0x+c_1), vg^2>$&4&1\\
\hline
$< u+vc_0, vg >$&5&1\\ 
\hline
$< u, v >$&6&1\\
\hline
$<g^2+uc_0g+vc_1g+uvc_2g >$ &1&3\\
\hline
$<g^2+uc_0g+vc_1g+uvc_2, uvg>$ &2&2\\
\hline
$<g^2+uc_0g+vc_1g, uv>$&3&1\\
\hline
$<g^2+uc_0g+vc_1+uvc_2, vg+uvc_3>$&2&2\\
\hline
$<g^2+uc_0g+vc_1, vg, uv>$&3&1\\
\hline
$<g^2+uc_0g, v>$&3&1\\
\hline
$<g^2+uc_0+vc_1g+uvc_2, ug+vc_3g+uvc_4 >$&2&2\\
\hline
$<g^2+uc_0+v(c_1g+c_0c_2), ug+vc_2g, uv>$&3&1\\
\hline
$<g^2+uvc_0, ug+uvc_1, vg+uvc_2>$&3&2\\
\hline
$<g^2+uc_0+vc_1, ug+vc_2, vg, uv>$, $c_0c_2=0$&4&1\\
\hline
$<g^2+uc_0, ug, v >$&4&1\\
\hline
$<g^2+vc_0g, u+v(c_1x+c_2), uv>$&4&1\\
\hline
$<g^2+vc_0, u+vc_1, vg >$&4&1\\
\hline
$<g^2, u, v >$&5&1\\
\hline
$<g+uc_0+vc_1+uvc_2>$&2&2\\
\hline
$<g+uc_0+vc_1, uv>$&3&1\\
\hline
$<g+uc_0,  v>$&3&1\\
\hline
$<g+vc_0, u+vc_1 >$&3&1\\
\hline
$<g, u, v>$ &4&1\\
\hline
$<1>$ & 3 & 1 \\
\hline
\end{tabular}
\begin{remark}
If $C$ is a cyclic code over a principal ideal ring $R$  of length $n$, then from \cite{Hori-Shiro99}, we know that the following inequality holds for the Hamming distance $d(C)$ of $C$:
\[ d(C) \leq n - \text{Rank}~C + 1.\]
Note That $R_{u^2,v^2,p}$ is a local ring but not principal ideal ring. From Table 1, it is clear that if $n$ is not relatively prime to $p$, then the above inequality does not hold for $R_{u^2,v^2,p}$. Hence, in general, the above inequality
may not hold for a local ring. It is good problem to establish analogous inequality for a local ring.
\end{remark}

In Table 2, we give examples of optimal ternary codes obtained as the Gray images of cyclic codes over $R_{u^2,v^2,3}$. In Table 2, $[~.~]^*$ will denote that the ternary code is an optimal code. From Table 2, we can see that we have obtained all ternary optimal code except $[12, 2, 9]^*$.\\

{\bf Table 2.}  Ternary images of some cyclic codes of length 3 over $R_{u^2,v^2,3}$.\\
\begin{center}
\begin{tabular}{| l | c| c |}
\hline
Non-zero generator polynomials & $\phi_L(C)$\\
\hline
$<uvg^2>$ & $[12, 1, 12]^*$\\
\hline
$< vg^2+uv2g >$& $[12, 2, 8]$\\
\hline
$< ug^2+vg^2+uv2g >$& $[12, 3, 8]^*$\\
\hline
$< ug^2+vg^2+uv, uvg>$& $[12, 4, 6]^*$\\
\hline
$< ug^2, vg^2, uvg>$&$[12, 5, 6]^*$\\
\hline
$< ug^2+vg, vg^2, uvg >$ & $[12, 6, 6]^*$\\
\hline
$< ug+vg, uv>$ & $[12, 7, 4]^*$\\
\hline
$< g^2, ug, vg >$ & $[12, 8, 3]^*$\\
\hline
$< g^2, ug, vg, uv >$ & $[12, 9, 3]^*$\\
\hline
$< g^2, ug, v >$ & $[12, 10, 2]^*$\\
\hline
$< g, u+2v >$ & $[12, 11, 2]^*$\\
\hline
\end{tabular}
\end{center}

\section{The images of cyclic codes over $R_{u^2,v^2,p}$} \label{p-ary-im}
Let $T$ be a cyclic shift operator defined as $T((c_0, c_1, \cdots, c_{n-1}))$ = $(c_{n-1}, c_0,$ $ \cdots, c_{n-2}).$  A linear code $C$ of length $n$ is said to be an $l$-quasi cyclic code if $T^l((c_0, c_1, \cdots, c_n)) \in C$ whenever $(c_0, c_1, \cdots, c_n) \in C$. Let $\phi_L$ be the Gray map on $R_{u^2,v^2,p}$ as defined in \ref{graymap}. We can prove the following lemma and theorem in a similar way to \cite[Lemma 5.2, Theorem 5.3]{Yil-Kar11}.
\begin{lemma}
 If $T$ is a cyclic shift operator on vectors of length $n$, then $\phi_L\circ T = T^4 \circ \phi_L$.
\end{lemma}
By using this lemma it is easy to prove the following theorem:
\begin{theo}
If $C$ is a cyclic code of length $n$ over $R_{u^2,v^2,p}$, then $\phi_L(C)$ is a $4$-quasi cyclic code of length $4n$ over $\Z_p$.
\end{theo}

\bibliographystyle{plain}
\bibliography{ref}

\end{document}